\documentclass[11pt]{article}
\usepackage{fullpage,amsfonts,amssymb,amsmath}
\usepackage{fullpage, amsfonts,amssymb,amsmath,amsthm,}

\newcounter {ctr1}
\newtheorem{mythm}[ctr1]{Theorem}
\newtheorem{mylem}[ctr1]{Lemma}
\newtheorem{myclm}[ctr1]{Claim}

\newtheorem*{hypA}{Hypothesis A (Randomized Optimization Hypothesis for $\Ppoly$)}
\newtheorem*{hypB}{Hypothesis B (Randomized Optimization Hypothesis for $\NCzero$)}

\newtheorem*{restatethmA}{Theorem~\ref{thm:amkcsp} (restated)}

\DeclareMathOperator{\Val}{Val}

\DeclareMathOperator{\poly}{poly}
\DeclareMathOperator{\Max}{Max}

\newcommand{\eps}{\varepsilon}

\newcommand{\bE}{\mathbb{E}}

\newcommand{\Polytime}{\mathsf{P}}
\newcommand{\NP}{\mathsf{NP}}
\newcommand{\coNP}{\mathsf{coNP}}

\newcommand{\AM}{\mathsf{AM}}

\newcommand{\MA}{\mathsf{MA}}

\newcommand{\PSPACE}{\mathsf{PSPACE}}

\newcommand{\Ppoly}{\mathsf{P/poly}}

\newcommand{\NCzero}{\mathsf{NC}^0}
\newcommand{\ACzero}{\mathsf{AC}^0}
\newcommand{\TCzero}{\mathsf{TC}^0}

\newcommand{\prAM}{\mathsf{prAM}}

\begin{document}

\title{A PCP Characterization of $\AM$}

\author{Andrew Drucker
\thanks{Email: adrucker@mit.edu.  Supported during part of this work by an Akamai Presidential Graduate Fellowship.}
\\MIT}

\maketitle

\begin{abstract}
We introduce a 2-round stochastic constraint-satisfaction problem, and show that its approximation version is complete for (the promise version of) the complexity class $\AM$.  This gives a `PCP characterization' of $\AM$ analogous to the PCP Theorem for $\NP$.  Similar characterizations have been given for higher levels of the Polynomial Hierarchy, and for $\PSPACE$; however, we suggest that the result for $\AM$ might be of particular significance for attempts to derandomize this class.

To test this notion, we pose some `Randomized Optimization Hypotheses' related to our stochastic CSPs that (in light of our result) would imply collapse results for $\AM$.  Unfortunately, the hypotheses appear over-strong, and we present evidence against them.  In the process we show that, if some language in $\NP$ is hard-on-average against circuits of size $2^{\Omega(n)}$, then there exist hard-on-average optimization problems of a particularly elegant form.

All our proofs use a powerful form of PCPs known as Probabilistically Checkable Proofs of Proximity, and demonstrate their versatility.  We also use known results on randomness-efficient soundness- and hardness-amplification.  In particular, we make essential use of the Impagliazzo-Wigderson generator; our analysis relies on a recent Chernoff-type theorem for expander walks. 

\end{abstract}



\section{Introduction}

\subsection{Background: PCPs and complexity classes}

A \emph{Constraint Satisfaction Problem (CSP)} is a collection $\psi(x)$ of Boolean-valued constraints over variables on a bounded-size alphabet $\Sigma$.  A CSP in which each constraint depends on at most $k$ variables is called a \emph{$k$-CSP}.  A natural computational task is to determine the maximum fraction of constraints that can be satisfied by any assignment.  Cook's Theorem~\cite{Pap} states that this problem is $\NP$-complete, and the landmark PCP Theorem of Arora et al.~\cite{ALM+} implies that, for a sufficiently small constant $\eps > 0$, it is $\NP$-hard even to output an estimate that is within $\eps$ of this maximum fraction (where in both results we may take $k = 3, \Sigma = \{0, 1\}$).

Given the importance of the PCP Theorem for complexity theory, researchers have looked for analogues of the result for complexity classes other than $\NP$.  The PCP Theorem can be seen as stating that it is $\NP$-hard to determine within $\eps$ the value of a 1-player `solitaire' game defined by a 3-CSP.  It is equally possible to study games played on a $k$-CSP in which 2 players alternate in setting values to designated blocks of variables, with one player trying to maximize the fraction of satisfied clauses and the other trying to minimize this fraction.  These games were explored in several works.  Ko and Lin~\cite{KL} showed that approximating the value of such a game is hard for the $j$-th level of the Polynomial Hierarchy, if the game lasts for $j$ moves.  In more recent work of Haviv et al.~\cite{HRT} this result was shown to hold even if each variable is allowed to appear in at most a constant number of constraints.

If the game is allowed to last polynomially many rounds, the approximation problem becomes $\PSPACE$-hard as shown by Condon et al.~\cite{CFLSa}.  The same authors showed the approximation problem for $\poly(n)$ rounds is also $\PSPACE$-hard if a maximizing player plays against a \emph{random} player~\cite{CFLSb} (where the game's value now is the \emph{expected} number of satisfied clauses under optimum play by the maximizer).  Moreover, all of the hardness-of-approximation results mentioned so far are in fact completeness results for the corresponding promise classes, so they can be viewed as giving `PCP characterizations' of $\NP, \PSPACE$, and the Polynomial Hierarchy.

One class that did not receive a PCP characterization based on CSP games was the Arthur-Merlin class $\AM$.  In fact, there are few known natural complete problems for $\AM$ (technically, for its promise version, $\prAM$; we don't know if $\AM$, a semantic class, has \emph{any} complete problems.  See Sec.~\ref{sec:promise} for the definition of $\prAM$.).  To this author's knowledge there is only one \emph{approximation} problem previously known to be $\prAM$-complete: Mossel and Umans~\cite{MU} give a $\prAM$-completeness result for approximating the VC dimension of set systems.  This striking result does not fall within the framework of CSP games given above.


\subsection{Our results}

In this paper we present a PCP characterization of $\prAM$.  We consider `stochastic' 2-CSPs $\psi(r, z)$, where $r$ is a collection of Boolean variables and $z$ a collection of variables over an alphabet $\Sigma$.  Let $\Val_{\psi}(r, z)$ be the fraction of constraints of $\psi$ satisfied by $(r, z)$.  In Section~\ref{pcpcharsec} we prove:

\begin{mythm}\label{thm:amkcsp} There is a finite alphabet $\Sigma$ and a constant $\eps > 0$, such that it is $\prAM$-complete to distinguish between the following two sets of $2$-CSPs:
$$
\Pi_{YES}=  \{ \psi: \text{ for all }r \text{ there exists }z \text{ such that } \Val_{\psi}(r, z) = 1\};  
$$
$$
\Pi_{NO}= \{ \psi:  \text{ with probability } 1 - \exp(-\Omega(|r|)) \text{ over random }r, \text{ } \Max_z [ \Val_{\psi}(r, z) ] < 1 - \eps  \}.      
$$
\end{mythm}
In particular, this implies that $\eps/2$-approximating the value of the 2-round game associated with $\psi(r, z)$ (where the first  player plays randomly) is $\AM$-hard.    

$\AM$ is a class for which we feel such a PCP characterization might be especially important.  There is compelling evidence that $\AM = \NP$, or at least that significant derandomization of $\AM$ is possible (see~\cite{SU} for an overview of this line of research).  One approach to to try and derandomize $\AM$ is to directly attack the `easiest' $\AM$-hard problems, and a problem like the one provided by Theorem~\ref{thm:amkcsp} seems like a plausible candidate.  

How might such an attack proceed?  We make a concrete suggestion in the form of two `Randomized Oracle Hypotheses'.  In what follows $\psi(r, z)$ is a 2-CSP over $\ell$ Boolean variables ($r$) and $m$ variables ($z$) over a finite alphabet $\Sigma$.


\begin{hypA} Fix any $\delta > 0$.  For every 2-CSP $\psi(r, z)$, there exists a circuit $C_{\psi}(r): \{0, 1\}^{\ell} \rightarrow \Sigma^m$ of size $O(\poly(|\psi|))$, such that with probability at least $1/\poly(\ell)$ over a random $r \in \{0, 1\}^{\ell}$, we have 
$$
 \Val_{\psi}(r, C_{\psi}(r)) \geq \Max_z [ \Val_{\psi}(r, z)  ]  - \delta .
$$
\end{hypA}

In a nutshell, this hypothesis suggests that it is easy to approximately-optimize over $z$ for a random choice of $r$, if we allow our optimizer to depend nonuniformly on the 2-CSP $\psi$.  (Such nonuniformity is clearly necessary, in light of the PCP Theorem for $\NP$.)  This hypothesis, if true, would yield a collapse result for $\AM$.  In Section~\ref{derandsec} we prove the following claim by a straightforward application of Theorem~\ref{thm:amkcsp}:

\begin{myclm}\label{clm:hypA}  Hypothesis A implies $\AM = \MA$.
\end{myclm}

A strengthened hypothesis could have the stronger implication that $\AM = \NP$.  Consider the following:

\begin{hypB}  For any $\delta > 0$, there is an integer $t = t(\delta) > 0$ such that the following holds.  For every 2-CSP $\psi(r, z)$, there exists a function $F_{\psi}(r): \{0, 1\}^{\ell} \rightarrow \Sigma^m$, where each output coordinate of $F_{\psi}$ depends on at most $t$ bits of $r$, and such that with probability at least $1 - \delta$ over $r$,
$$
\Val_{\psi}(r, F_{\psi}(r)) \geq \Max_z [ \Val_{\psi}(r, z)  ]  - \delta .
$$
\end{hypB}

In Section~\ref{derandsec} we prove:

\begin{myclm}\label{clm:hypB}  Hypothesis B implies $\AM = \NP$.
\end{myclm}

Given the potential consequences of these hypotheses, what chance do they have of being true?  Unfortunately, it seems that each is unlikely.  In Section~\ref{nphardsec} we prove two results, each to the effect that, if $\NP$ decision problems are hard on average for exponential-size circuits, then both hypotheses fail in a strong way.  We state these results next.  A language $L$ is called \emph{$p(n)$-hard for size $s(n)$} if for every circuit $C$ of size $s(n)$, $\Pr_{x \in \{0, 1\}^n}[C(x) = L(x)] \leq p(n)$.

\begin{mythm}\label{thm:avghard1}  Suppose there exists a $\gamma_1 > 0$ and an $L \in \NP \cap \coNP$ that is $(1 - 1/\poly(n))$-hard for size $2^{\gamma_1 n}$.
Then there exists $c, \gamma_2, \theta > 0$ and a polynomial-time constructible family $\{\psi_n(r, w)\}_{n > 0}$ of 2-CSPs (with $|r| = cn, |w| = d(n) = O(\poly(n))$), such that:

(\ref{thm:avghard1}.i) for all $r$, there exists a $w$ such that $\Val_{\psi_n}(r, w) = 1$;

(\ref{thm:avghard1}.ii) for all $n$, if $C: \{0, 1\}^{cn} \rightarrow \{0, 1\}^{d(n)}$ is a circuit of size at most $2^{\gamma_2 n}$, then 
$$
\Pr_{r} [  \Val_{\psi_n}(r, C(r))  >  1  -   \theta  ]   \leq \exp \{-\Omega(n)\}. 
$$
\end{mythm}

\vspace{.5 em}

\begin{mythm}\label{thm:avghard2}  There is an $\eps_0 > 0$ such that the following holds.  Suppose there exists a $\gamma_1 > 0$ and an $L \in \NP$ that is $(1/2 + \eps_0)$-hard for size $2^{\gamma_1 n}$.
Then there exists a $c > 0$, a polynomial-time constructible family $\{\psi_n(r, w)\}_{n > 0}$ of 2-CSPs (with $|r| = cn, |w| = d(n) = O(\poly(n))$), and $\gamma_2, \theta > 0$, such that: 

(\ref{thm:avghard2}.i) With probability $\geq 1 - \exp\{-\Omega(n)\}$ over $r$, there exists $w$ with $\Val_{\psi_n}(r, w) = 1$;

\vspace{.5 em}

(\ref{thm:avghard2}.ii) If $C: \{0, 1\}^{c(n)} \rightarrow \{0, 1\}^{d(n)}$ is any circuit of size at most $2^{\gamma_2 n}$, then 
$$
\Pr_{r} [  \Val_{\psi_n}(r, C(r))  >  1  -   \theta  ]   \leq \exp \{-\Omega(n)\}.
$$
\end{mythm}

We note that the hypothesis in Theorem~\ref{thm:avghard2} is implied by the hypothesis that there exists a \emph{balanced} function $L \in \NP$ that is $(1 - 1/\poly(n))$-hard for some size $s(n) = 2^{\Omega(n)}$; this follows from a result of O'Donnell~\cite{OD} (see also Healy et al.~\cite{HVV}, where the needed form of O'Donnell's result is made explicit and proved in a stronger form).
 
Theorems~\ref{thm:avghard1} and~\ref{thm:avghard2} both say that if $\NP$ (or $\NP \cap \coNP$) has sufficiently hard problems, then this hardness can be `concentrated' into a kind of `inapproximability-on-average' result for an optimization problem associated with a single, uniform family of (stochastic) CSPs.  Note that the two results offer a tradeoff: Theorem~\ref{thm:avghard2} gives a slightly weaker conclusion from a presumably likelier hardness assumption.  The assumptions in the above results are strong but, we feel, plausible.  But at the very least, these results suggest that the approach we suggested to showing new upper-bounds on the power of $\AM$ must be modified to have a reasonable chance of succeeding.  

We feel, however, that the Random Optimization Hypotheses are worthy of study in their own right, even if they turn out to be false; we pose some concrete questions about them at the end of the paper.  We also feel that Theorems~\ref{thm:avghard1} and~\ref{thm:avghard2} are interesting for the study of average-case hardness in $\NP$, and that the CSP families they produce might have further applications in complexity theory.

\subsection{Our methods}
All of our three main results--Theorems~\ref{thm:amkcsp}, \ref{thm:avghard1}, and~\ref{thm:avghard2}--are essentially hardness results for computational tasks associated with 2-CSPs.  In each case the reduction with which we prove our result uses a powerful type of PCP known as Probabilistically Checkable Proofs of Proximity (PCPPs). PCPPs were introduced independently by Ben-Sasson et al.~\cite{BGH+} and by Dinur and Reingold~\cite{DR}, and the PCPPs we use were developed by Dinur~\cite{Din} (in~\cite{DR, Din} PCPPs are referred to as `Assignment Testers').  In Section~\ref{sec:augpcpp}, we derive a variant form of PCPPs (Lemma~\ref{lem:searchcsp}) that is more useful for our purposes.

Lemma~\ref{lem:searchcsp} gives a general reduction (similar to past uses of PCPPs, e.g., in~\cite{Din}) in which we start with a two-argument circuit $Q(r, w)$ and efficiently produce a 2-CSP $\psi(r, z)$.  The basic hope for our reduction is as follows: first, for any $r$, if the restricted circuit $Q(r, \cdot)$ is satisfiable (i.e., there exists $w$ such that $Q(r, w) = 1$), then the restricted 2-CSP $\psi(r, \cdot)$ should be satisfiable as well.  Second, if $Q(r, \cdot)$ is unsatisfiable, then any assignment to $\psi(r, \cdot)$ should violate an $\Omega(1)$-fraction of the constraints in $\psi$.  Unfortunately, this second requirement is too strong and cannot be met.  What we \emph{can} guarantee is that if $r$ is `far' in Hamming distance from any $r'$ for which $Q(r', \cdot)$ is satisfiable, then for any $z$, $(r, z)$ violates an $\Omega(1)$-fraction of constraints of $\psi$.

How does this reduction help us prove the $\prAM$-hardness result in Theorem~\ref{thm:amkcsp}?  Any instance $x$ of a promise problem $\Pi = (\Pi_{YES}, \Pi_{NO})$ defines a predicate $Q(r, w)$ computed by a poly-size circuit.  If $x \in \Pi_{YES}$ then for all $r$, $Q(r, \cdot)$ is satisfiable; while if $x \in \Pi_{NO}$ then for a 2/3 fraction of $r$, $Q(r, \cdot)$ is unsatisfiable.  In order to apply our reduction, we need a stronger condition in the second case: a random choice of $r$ should be far from any $r'$ for which $Q(r', \cdot)$ is satisfiable.  In other words, we need an extremely low error probability in our underlying Arthur-Merlin protocol.  This cannot be achieved by straightforward parallel repetition, but it is provided by a theorem of Bellare et al.~\cite{BGG} which gives a randomness-efficient soundness-amplification for $\AM$.  Interestingly, Mossel and Umans~\cite{MU} also used such amplification for their $\AM$-hardness-of-approximation result on VC dimension, but for rather different reasons (unrelated to PCPs).

Our $\prAM$-hardness proof is, we feel, more straightforward than the existing proofs of the analogous results for $\PSPACE$ and the Polynomial Hierarchy, modulo our use of sophisticated tools (PCPPs and efficient soundness-amplification) which we apply in a `black-box' fashion.  Of course we do not rule out that our result could be also proved more directly by adapting ideas from the earlier papers (which use some of the same property-testing ideas that have gone into constructions of PCPPs).  But we feel that PCPPs in particular, which have already found applications in PCP construction, coding theory, and property testing (see~\cite{BHLM} for an overview), are a versatile tool which could be of more widespread use in complexity theory.  In a very recent example of their utility, Williams~\cite{Wil10} applied PCPPs to the study of circuit lower bounds.

Next we discuss our methods in Theorems~\ref{thm:avghard1} and~\ref{thm:avghard2}.  Our transformation in Lemma~\ref{lem:searchcsp} from the circuit $Q$ to the 2-CSP $\psi$ has a further useful property: we can reduce the problem of \emph{finding} satisfying assignments to $Q(r, \cdot)$, to the problem of finding nearly-optimal assignments to $\psi(r, \cdot)$.  Roughly speaking, we show the following.  Suppose there is an algorithm $P(r)$ producing an assignment $z$, such that with some probability $p$ over $r$, $(r, P(r))$ satisfies `almost all' of the constraints of $\psi$; then there is a second algorithm $\tilde{P}(r)$ such that $Q(r, \tilde{P}(r)) = 1$ with probability $p' \geq 2^{-\eps |r|}p$ (where $\eps > 0$ can be chosen arbitrarily small).  This property of the reduction is somewhat more novel, although the techniques we use (involving error-correcting codes) still follow previous works.

To apply our reduction, we use the hardness assumptions in Theorems~\ref{thm:avghard1} and~\ref{thm:avghard2} to produce predicates $Q(r, w)$ such that $Q(r, \cdot)$ is satisfiable with high probability, while any `small' witness-producing circuit $C$ fails to solve the search problem associated with $Q$: that is, $Q(r, C(r)) = 0$ with high probability.  Because of the exponential loss factor $2^{-\eps |r|}$ in our reduction, we need the search problem associated with $Q$ to be \emph{extremely} hard: we need every `small' circuit $C$ to succeed with probability at most $\exp\{-\Omega(|r|)\}$ over $r$ in achieving $Q(r, C(r)) = 1$.  

To produce such extremely hard search problems from a more `mild' hardness assumption, we use existing hardness-amplification techniques.   In particular, we use the well-known Impagliazzo-Wigderson generator~\cite{IW}.  This generator, on input parameter $n$, takes a seed $r$ of length $O(n)$, and produces $n$ `pseudorandom' outputs $g_1, \ldots, g_n$ each of length $n$.  The generator has the property that if language $L$ is mildly hard for sufficiently small (but exponential-size) circuits, then any sufficiently smaller circuit has success probability $ \leq \exp\{-\Omega(n)\}$ in correctly guessing the $n$-bit string $(L(g_1), \ldots, L(g_n))$.
Then, if our hard language $L$ is in $\NP \cap \coNP$ (as in Theorem~\ref{thm:avghard1}), defining our predicate $Q$ is straightforward: we let $Q(r, w) = 1$ iff $w$ contains `proofs' for the $n$ values $(L(g_1), \ldots, L(g_n))$.  

If our hard language is merely in $\NP$ (as in Theorem~\ref{thm:avghard2}), we need to work harder.  In this case, we let $Q(r, w) = 1$ iff $w$ contains proofs that $L(g_i) = 1$, for a `sufficient number' of the strings $g_i$.  The idea is that if a small circuit $C(r)$ could with some noticeable probability guess such proofs for `almost all' the indices $i$ for which $L(g_i) = 1$, then $C$ could be modified to correctly guess $(L(g_1), \ldots, L(g_n))$ with noticeable probability, contrary to the properties of the generator.  Making this idea work involves showing that the set size $|\{ i \in [n]: L(g_i) = 1\}|$ is highly concentrated around its expectation.  For this we rely on a recently proved concentration result called the Strong Chernoff Bound for Expander Walks~\cite{WX05, WX08, Hea}.  This result is perfectly suited to analyze the Impagliazzo-Wigderson generator (which is partly defined in terms of walks on expander graphs).

The precise form of our assumptions in Theorems~\ref{thm:avghard1} and~\ref{thm:avghard2} are dictated by the hardness-amplification tools currently available.  In particular, sufficiently strong hardness-amplification is only available if we make a hardness assumption against nonuniform, exponential-sized circuits.  We believe versions of Theorems~\ref{thm:avghard1} and~\ref{thm:avghard2} should be possible for a \emph{uniform} hardness assumption; recently Impagliazzo et al.~\cite{IJKW} made partial progress towards the hardness-amplification tools needed.

\section{Preliminaries}\label{prelimsec}

\subsection{Basic definitions}\label{basicsec}

We presume familiarity with basic notions in complexity theory, in particular familiarity with the classes $\Polytime, \NP$, and $\AM$.  We define promise classes and the promise class $\prAM$ in Section~\ref{sec:promise}.

For a language $L \subseteq \{0, 1\}^*$, we use $L(x)$ to denote the characteristic function of $L$.
We use $|x|$ to denote the length of a string $x$ over some (possibly non-Boolean) alphabet $\Sigma$.  $d(x, y)$ denotes the Hamming distance between strings $x, y \in \Sigma^n$, and $d(x, S)$ is the generalized Hamming distance between $x \in \Sigma^n$ and a set $S \subseteq \Sigma^n$.  If $d(x, S) \leq c$ we say $x$ is \emph{$c$-close to $S$}, otherwise $x$ is \emph{$c$-far from $S$}.  Similarly, for $\alpha \in [0, 1]$, if $d(x, S) \leq \alpha n$ we say $x$ is \emph{$ \alpha$-close in relative distance to $S$}, otherwise $x$ is \emph{$ \alpha$-far in relative distance from $S$}.

$H(t): [0, 1] \rightarrow [0, 1]$ denotes the binary entropy function, $H(t) = -t \log t - (1 - t)\log (1 - t)$ for $t \in (0, 1)$ and $H(0) = H(1) = 0$.  We let $V_{n, k}$ denote the discrete volume of the Hamming sphere of radius $k$ in $\{0, 1\}^n$; that is, 
$$
V_{n, k} := \sum_{0 \leq i \leq k} {n \choose i},
$$ 
and we use the known bound $V_{n, \alpha n} \leq 2^{H(\alpha)n}$ (valid for $\alpha \in [0, 1/2]$).

When we speak of circuits, unless otherwise mentioned we mean deterministic Boolean circuits of fanin-two, and we measure circuit size (denoted $|C|$ for circuit $C$) as the number of gates.  For functions $p(n) \in [0, 1]$, $s(n) \geq 0$ we say that a language $L$ is \emph{$p(n)$-hard for size $s(n)$} if for every Boolean circuit $C$ of size $\leq s(n)$, $\Pr_{x \in \{0, 1\}^n}[C(x) = L(x)] \leq p(n)$.  We extend this definition to general functions: we say that a function $F: \{0, 1\}^n \rightarrow \{0, 1\}^m$ is $p(n)$-hard for size $s(n)$ if for every $m$-output Boolean circuit $C$ of size $\leq s(n)$, $\Pr_{x \in \{0, 1\}^n}[C(x) = F(x)] \leq p(n)$.

\subsection{CSPs, PCPPs, and codes}

Fix an integer $k \geq 1$.  A \emph{$k$-local Constraint Satisfaction Problem}, or \emph{$k$-CSP}, over finite alphabet $\Sigma$ is a collection $\psi(x) = \psi_1(x), \ldots \psi_m(x)$ of Boolean-valued functions on the input $x = (x_1, \ldots x_n) \in \Sigma^n$, where each $\psi_j$ depends only on some $k$ variables of $x$ and is specified by a $k$-tuple $I_j \subseteq [n]$ and a truth-table on these $k$ variables.  Define $\Val_{\psi}(x)$, the \emph{value of $\psi$ on $x$}, as the fraction of constraints $\psi_j$ that are satisfied by $x$ (i.e. such that $\psi_j(x) = 1$).

Next we define PCPPs.  Fix a circuit $C(x)$ on $n$ Boolean input variables, a finite alphabet $\Sigma$, and a parameter $\beta > 0$.  We say that a $k$-CSP $\psi$ is a \emph{PCPP for $C$ over $\Sigma$ with security $\beta$} if:

\begin{enumerate}
\item $\psi$ is defined on variable set $(x, z)$, where $x$ are the Boolean input variables to $C$ and $z$ are auxiliary `proof' variables taking values in $\Sigma$;
\item For any $x \in \{0, 1\}^n$, if $C(x) = 1$ then there exists a setting of $z$ such that $Val_{\psi}(x, z) = 1$;
\item For all $x  \in \{0, 1\}^n$ and $z$, $\Val_{\psi}(x, z) \leq 1 - \beta\cdot \frac{d(x, C^{-1}(1))}{n}$. 
\end{enumerate}
The \emph{proof size} of $\psi$ is the number of variables in $z$.

The following positive result on PCPPs is due to Dinur.  

\begin{mythm}\label{thm:pcpp}~\cite[Cor. 9.3]{Din}
There is a constant-size alphabet $\Sigma_0$, a constant $\beta > 0$, and a polynomial-time algorithm that, given a circuit $Q(x)$ of size $t$, produces a 2-CSP $\psi_Q(x, z)$ that is a PCPP for $Q$ over $\Sigma_0$ with security $\beta$.  Moreover, the proof size of $\psi$ is $O(\poly(t))$.
\end{mythm}

Following the techniques of earlier papers working with PCPPs, we will use PCPPs in conjunction with efficient error-correcting codes.  A (binary) \emph{code} is an injective map $E: \{0, 1\}^N \rightarrow \{0, 1\}^{N'}$ where $N' \geq N$.  We also use $E$ to denote the image of the map, i.e., we consider $E \subseteq \{0, 1\}^{N'}$.  The \emph{minimum distance} of the code is the minimum over distinct $u, v \in E$ of $d(u, v)$.  An algorithm $D$ \emph{decodes $E$ from an $\eta$ fraction of errors} if, given any string $u$ at relative distance at most $\eta$ from some $u' \in E$, $A(u)$ outputs $u'$.  Note that for such decoding to be possible, the minimum distance must be greater than $2\eta N'$.

We will use the following well-known fact: 

\begin{mythm}\label{thm:goodcodes}
There is a polynomial-time computable code $E$ for all input lengths $N$ with output length $N' = O(N)$, and an $\eta > 0$, such that $E$ can be polynomial-time decoded from an $\eta$ fraction of errors.
\end{mythm}

Many such constructions are known; recently Goldwasser et al.~\cite{GGH+} gave a construction in which the decoder algorithm can be implemented in $\ACzero$, i.e., with constant-depth, polynomial-size Boolean circuits.

\subsection{Promise problems and $\prAM$}\label{sec:promise}

A \emph{promise problem} is a pair $\Pi = (\Pi_{YES}, \Pi_{NO})$ of disjoint subsets of $\{0,1\}^*$ (the `yes' and `no' instances, respectively).  For a function $s(n) \in [0, 1]$, we say that $(\Pi_{YES}, \Pi_{NO}) \in \prAM_{1, s(n)}$ if there exists a polynomial-time randomized algorithm $M(x, r, w)$, with $|r|, |w| = O(\poly(n))$ such that:
\begin{enumerate}
\item (Completeness) If $x \in \Pi_{YES}$, then with probability $1$ over the random string $r$, there exists a $w = w(r)$ such that $M(x, r, w) = 1$;
\item (Soundness) If $x \in \Pi_{NO}$ and $|x| = n$, then the probability over the random string $r$ that there exists a $w$ such that $M(x, r, w) = 1$, is at most $s(n)$.
\end{enumerate}
The algorithm $M$ defines an `Arthur-Merlin protocol': we consider that a polynomially bounded verifier Arthur chooses a random `challenge' $r$ for the computationally unbounded Merlin, who sees $r$ and gives a response $w$ which Arthur accepts or rejects..

We define $\prAM = \prAM_{1, 1/3}$.  A promise problem $\Pi_1 = (\Pi_{YES}, \Pi_{NO})$ is \emph{$\prAM$-hard} if for all $\Pi' = (\Pi'_{YES}, \Pi'_{NO})$ in $\prAM$, there exists a polynomial-time computable reduction $R(x)$, such that, if $x \in \Pi'_{YES}$, then $R(x) \in \Pi_{YES}$, while if $x \in \Pi'_{NO}$, then $R(x) \in \Pi_{NO}$.  We say that $\Pi$ is \emph{$\prAM$-complete} if $\Pi$ is in $\prAM$ and is $\prAM$-hard.

It is not hard to see that for any $\Pi \in \prAM$, the soundness parameter $1/3$ in the protocol can be made exponentially small in $n$, by parallel repetition of the original protocol.  However, we require soundness-amplification that is more efficient in its use of randomness.  This is provided by a result of Bellare et al.~\cite{BGG}.  They state their theorem for $\AM$, not for $\prAM$, but the proof carries over without changes to the promise setting and we state it for this setting.

\begin{mythm}\label{thm:amp}~\cite{BGG}
Let $\Pi = (\Pi_{YES}, \Pi_{NO}) \in \prAM$, where $M(x, r, w)$ is a polynomial-time predicate defining an Arthur-Merlin protocol for $\Pi$.  Let $n = |x|$, and fix a polynomial $m(n)$.  Then there exists an Arthur-Merlin protocol for $\Pi$ defined by a polynomial-time predicate $M'(x, r', w')$, with $|w'| \leq O(\poly(n)), |r| \leq |r'| \leq O(|r| + m(n))$, and with soundness $2^{-m(n)}$.
\end{mythm}

The randomness-efficiency in the above result has been improved in more recent work (see~\cite{MU} for a discussion), but we do not need or use these improvements.

\subsection{AM-$k$-CSPs}\label{sec:amcspdef}

By an \emph{AM-$k$-CSP} we mean a $k$-CSP $\psi(r, z)$, where $r$ are Boolean and $z$ may be non-Boolean.  We call $r$ the `Arthur-variables' and $z$ the `Merlin-variables'.  Informally speaking, we are interested in the game in which the $r$ are first set uniformly by Arthur, and then Merlin sets $z$ to try to maximize the fraction of constraints of $\psi$ satisfied by $(r, z)$.

For any fixed $k \geq 1$, soundness parameter $s = s(|r|) \in [0, 1]$, alphabet $\Sigma$, and fixed $\eps \in (0, 1]$, we define the promise problem $\mathsf{Gap-AM-k-CSP}_{1, 1 - \eps, s(|r|)} = (\Pi_{AM-CSP, YES}, \Pi_{AM-CSP, NO})$ as follows.  Both `yes' and `no' instances are AM-$k$-CSPs over $\Sigma$.
If $\psi(r, z) \in\Pi_{AM-CSP, YES}$, we are promised that for all choices of $r$, there exists a $z$ such that $\Val_{\psi}(r, z) = 1$.  
If $\psi(r, z) \in \Pi_{AM-CSP, NO}$, we are promised that only for at most an $s(|r|)$ fraction of strings $r$ does there exist a $z$ with $\Val_{\psi}(r, z) > 1 - \eps$.

\section{An Augmented PCPP}\label{sec:augpcpp}

As a tool for proving Theorems~\ref{thm:amkcsp}, \ref{thm:avghard1}, and~\ref{thm:avghard2}, we prove the following `augmented' version of the PCPP Theorem (Theorem~\ref{thm:pcpp}), which we derive from Theorem~\ref{thm:pcpp}.  We remark that the proof of our $\prAM$-completeness result (Theorem~\ref{thm:amkcsp}) uses only condition (\ref{lem:searchcsp}.i) of the Lemma below; this first part of the Lemma is quite similar to previous uses of PCPPs.  Also, our use of error-correcting codes will only be important for establishing condition (\ref{lem:searchcsp}.ii).

\begin{mylem}\label{lem:searchcsp} 
There is a finite alphabet $\Sigma_0$ such that the following holds.  For any $\eps > 0$ there is a $\nu > 0$ and a polynomial-time algorithm $A$ that takes as input a Boolean circuit $C = C(r, w)$.  $A$ outputs a 2-CSP $\psi(r, z)$, where $|z| = O(\poly(|C|))$ and the variables of $z$ are over $\Sigma_0$.  Letting $\ell = |r|$, $\psi$ has the following properties:

\begin{itemize}
\item[(\ref{lem:searchcsp}.i)] For all $r$, if there is a $w$ such that $C(r, w) = 1$, then there is a $z$ such that $\Val_{\psi}(r, z) = 1$.  On the other hand, if $r$ is $\alpha \ell$-far from any $r'$ for which $C(r', \cdot)$ is satisfiable, then for all $z$, $\Val_{\psi}(r, z) < 1 - \Omega(\alpha)$.

\item[(\ref{lem:searchcsp}.ii)] Suppose $P(r)$ is any (possibly randomized) procedure such that with probability at least $p = p(\ell)$ over a uniform $r \in \{0,1\}^{\ell}$ and any randomness in $P$, $P(r)$ outputs a $z$ such that $\Val_{\psi}(r, z) > 1 - \nu$.

Then there exists a deterministic procedure $\tilde{P}(r)$, such that with probability at least $p(\ell)\cdot 2^{- \eps \ell}$ over uniform $r$, $\tilde{P}(r)$ outputs a $w$ such that $C(r, w) = 1$.
Moreover, $\tilde{P}(r)$ is computable by a nonuniform, $\poly(|C|)$-sized circuit that makes a single oracle call to $P$ on the same input length (with $P$'s randomness fixed nonuniformly).
\end{itemize} 

\end{mylem}

\begin{proof}

Let $N := |w|$.  We can assume, by padding $w$ if necessary, that $N \geq \ell$.  Let $E: \{0,1\}^{N} \rightarrow \{0, 1\}^{N'}$ be the error-correcting code given by Theorem~\ref{thm:goodcodes} (with $N' = O(N)$), applied to inputs of length $N$.

Let $b = b(\ell) := \lceil \frac{N'}{\ell}  \rceil$.  Define a predicate $Q(r_1, r_2, \ldots, r_b, u)$, with $|r_i| = |r| = \ell$ for $i \leq b$ and $|u| = N'$, by the following rule: $Q(r_1, r_2, \ldots r_b, u) = 1$ iff $r_1 = r_2 = \ldots, = r_b$, $u = E(w)$ for some $w$, and $C(r_1, w) = 1$.  Clearly we can efficiently construct a circuit of size $O(\poly(|C|))$ computing $Q$.  Note that by our setting of $b$, there are more variables in the blocks $r_j$ than in $u$.

Let $\psi_0 = \psi_Q((r_1, \ldots, r_b, u), Z)$ be the PCPP 2-CSP for $Q$ over alphabet $\Sigma_0$ given by Theorem~\ref{thm:pcpp}, efficiently constructible and of size $O(\poly(|C|))$.  We take $\psi_0$ and make two changes.  First, we substitute the variables of $r$ for the corresponding variables of each vector $r_i$.  Second, we allow the variables of $u$ to range over all of $\Sigma_0$ (we may assume $\{0, 1\} \subseteq \Sigma_0$), and modify each constraint to reject in the case where one or more of its $u$-variables are set to a non-Boolean value.

We denote the resulting 2-CSP by $\psi(r, u, Z)$.  We claim that this efficiently constructible 2-CSP satisfies the conditions of Lemma \ref{lem:searchcsp}'s statement, with $z := (u, Z)$ and $\Sigma_0$ as in Theorem~\ref{thm:pcpp}.  Our setting of $\nu > 0$ will be determined later.

First, we show that condition (\ref{lem:searchcsp}.i) is satisfied.  Consider any $r \in \{0, 1\}^{\ell}$.  Suppose that there exists $w \in \{0, 1\}^N$ such that $C(r, w) = 1$.  Then $Q(r, r, \ldots , r, E(w)) = 1$.  Using the completeness property of PCPPs, there exists a $Z$ such that $\Val_{\psi}((r, E(w)), Z) = 1$.  On the other hand, say $r$ is $\alpha \ell$-far from any $r'$ for which $C(r', \cdot)$ is satisfiable.  Given any $u \in \Sigma_0^{N'}$, let  us choose some string $u' \in \{0, 1\}^{N'}$ which agrees with $u$ on any variable where $u$ is Boolean.  We observe that $(r, r, \ldots , r, u')$ is $\alpha/ 2$-far in relative distance from $Q^{-1}(1)$.  By the soundness property of PCPPs, for any choice of $Z$, $\Val_{\psi}((r, u'), Z) < 1 - \frac{\alpha \beta }{2}$, where $\beta > 0$ is the constant from Theorem~\ref{thm:pcpp}.  Also, by the way we defined $\psi$, $\Val_{\psi}((r, u), Z) \leq \Val_{\psi}((r, u'), Z)$.  We have verified condition (\ref{lem:searchcsp}.i).

\vspace{.5 em}

Now we turn to condition (\ref{lem:searchcsp}.ii).  Let $P(r)$ be as described in (\ref{lem:searchcsp}.ii).  Note that  by averaging, we may fix (nonuniformly) some value of the randomness used by $P$ while preserving the lower-bound $p(\ell)$ on its success probability over the choice of $r$; we do so and consider $P$ a deterministic algorithm from now on.  We set $ \nu := \eta \beta \gamma/4$, where $\eta$ is the constant in Theorem~\ref {thm:goodcodes}, $\beta$ is the constant in Theorem~\ref{thm:pcpp}, and $\gamma \in (0,1)$ is a small constant to be announced.

Let $P'(r)$ be the procedure that, on input $r$, computes $z = P(r) = (u, Z)$ and runs the polynomial-time decoder for $E$ on $u$, yielding a string $w \in \{0,1\}^{N}$.  Let $P'$ output $w$.

We analyze the behavior of $P'$.  Let $z = (u, Z)$ be any output of $P(r)$ such that $\Val_{\psi}(r, z) > 1 - \nu$.  By the soundness property of PCPPs, the string $(r, r, \ldots, r, u)$ must be $\frac{\nu}{\beta} = \frac{\eta \gamma}{4}$-close in relative distance to some string $(r_1, \ldots, r_b, u')$ for which $Q(r_1, \ldots, r_b, u') = 1$ (and thus $r_1 = \ldots = r_b$ and $u' \in E$).  Since $|u'| = N' \leq b(\ell)\cdot |r| \leq 2N'$, we find that $d(r, r_1) < \gamma \ell$ and $d(u, u') < \eta N'$.  The latter inequality implies that when $P'$ applies the polynomial-time decoder to $u$, it correctly recovers $w = E^{-1}(u')$.  Since $Q(r_1, \ldots, r_b, u') = 1$, we have $C(r_1, w) = 1$.

To analyze $\tilde{P}$, say that a string $r \in \{0, 1\}^n$ is \emph{good} if $P'(r)$ outputs a $w$ such that there exists an $r'$ at distance at most $\gamma \ell$ from $r$, such that $C(r', w) = 1$.  Our analysis of $P'$, combined with our original assumption about the success probability of $P$, shows that at least a $p(\ell)$ fraction of strings $r$ are good.

Now we define the procedure $\tilde{P}(r)$: $\tilde{P}(r)$ first chooses a vector $v \in \{0,1\}^l$ uniformly from the set of all strings of Hamming weight at most $\gamma l$, then outputs $P'(r + v)$.  Note that, if $r$ is selected uniformly, $r + v$ is also uniform and, after conditioning on its value, $r$ is uniformly distributed over all strings at distance at most $\gamma \ell$ from $r + v$.  Thus, conditioning on $r + v$ being good, we have at least a $1/V_{\ell, \gamma \ell} \geq 2^{-H(\gamma)\ell}$ chance that $C(r, P'(r + v)) = 1$.  So the overall success probability of $\tilde{P}(r)$ is at least $p(\ell)\cdot 2^{-H(\gamma)\ell}$.  Since $H(\gamma) \rightarrow 0$ as $\gamma \rightarrow 0$, we may choose $\gamma > 0$ so that the success probability is at least $p(\ell)\cdot 2^{-\eps \ell}$.

$P'$ is clearly a polynomial-time algorithm making one call to $P$, while $\tilde{P}$ simply makes one call to $P'$ after its random sampling and bitwise addition mod 2.  The choice of $v$ can be nonuniformly fixed in a way that does not decrease the success probability, so $\tilde{P}$ can be implemented with the resources claimed.  Thus we have verified condition (\ref{lem:searchcsp}.ii), completing the proof of the Lemma.
\end{proof}

\section{PCP Characterization of $\prAM$}\label{pcpcharsec}
In this section we prove Theorem~\ref{thm:amkcsp}, which we restate in the terminology of Section~\ref{sec:amcspdef}:
\begin{restatethmA} There is a finite alphabet $\Sigma$ and a constant $\eps > 0$, such that 
\\*$\mathsf{Gap-AM-2-CSP}_{1, 1 - \eps, \exp\{-\Omega(|r|)\}}$ is $\prAM$-complete.
\end{restatethmA}

\begin{proof} 
First, we claim that for any $s(|r|) = o(1)$ and $\eps > 0$, $\mathsf{Gap-AM-2-CSP}_{1, 1 - \eps, s(|r|)} =$\\$\left( \Pi_{AM-CSP, YES}, \Pi_{AM-CSP, NO} \right)$ 
is in $\prAM$.  The protocol is as follows: given a 2-CSP $\psi(r, z)$, Arthur picks $r$ uniformly and Merlin responds with a setting of $z$.  Arthur accepts iff $\Val_{\psi}(r, z) = 1$.  If $\psi \in \Pi_{AM-CSP, YES}$, then clearly Arthur accepts with probability 1 when Merlin responds optimally.  If $\psi \in \Pi_{AM-CSP, NO}$, then Arthur accepts with probability at most $s(|r|)$, which is greater than $2/3$ for large enough $|r|$.  (For instances with $|r|$ below this threshold, Arthur can simply request certificates $z(r)$ for every setting of $r$ and verify that each satisfies $\Val_{\psi}(r, z(r)) = 1$.)
 
Thus our main task is to show that the promise problem is $\prAM$-hard, for appropriate choice of parameters.  Let $\Pi = (\Pi_{YES}, \Pi_{NO}) \in \prAM$, and let $M_1(x, r_1, w_1)$ be a polynomial-time-computable predicate defining an Arthur-Merlin protocol for $\Pi$.  We use parameters $n = |x|,  \ell_1(n) = |r|$; by definition of $\prAM$ we have $\ell_1(n) = O(\poly(n))$ and $|w_1| = O(\poly(n))$.  By padding $r_1$ if necessary we may assume $\ell_1(n) \geq n$.  Apply Theorem~\ref{thm:amp} to $M_1$, with the setting $m(n) := \ell_1(n)$.  Thus we get a new Arthur-Merlin protocol $M_2(x, r_2, w_2)$ for $\Pi$, with $|r_2| = \ell_2(n) \in [n, \ldots, D\cdot \ell_1(n)  ]$ (for some fixed $D > 0$), $|w_2| = O(\poly(n))$, and with soundness $2^{-\ell_1(n)}$.

Given an input $x \in \Pi_{YES} \cup \Pi_{NO}$, we construct a $\poly (n)$-sized circuit $C(r_2, w_2) = C_x(r_2, w_2)$ that accepts iff $M_2(x, r_2, w_2) = 1$.  To this circuit we apply the algorithm $A$ of Lemma~\ref{lem:searchcsp} (with a setting of $\eps > 0$ to be announced), yielding a 2-CSP $\psi = \psi(r_2, z)$ which we make the output of our reduction.

We show the correctness of the reduction.  First, suppose that $x \in \Pi_{YES}$.  Then for each choice of $r_2$, there exists a $w_2$ such that $M_2(x, r_2, w_2) = 1$.  By condition (\ref{lem:searchcsp}.i) of Lemma~\ref{lem:searchcsp}, there exists $z$ such that $\Val_{\psi}(r_2, z) = 1$.  Thus $\psi \in \Pi_{AM-CSP, YES}$.

Now suppose that $x \in \Pi_{NO}$.  Then by the soundness property of $M_2$, the number of strings $r_2$ for which $M_2(x, r_2, \cdot)$ is satisfiable is at most $2^{-\ell_1(n)}\cdot 2^{\ell_2(n)} \leq 2^{(1 - \frac{1}{D})\ell_2(n)}$.  Thus the number of $r_2$ for which there exists an $r'$ at distance $\leq \alpha \ell_2(n)$ from $r_2$, such that $M_2(x, r', \cdot)$ is satisfiable, is at most 
$$
V_{\ell_2(n), \alpha \ell_2(n)}\cdot 2^{(1 - \frac{1}{D})\ell_2(n)}   \leq 2^{(H(\alpha) + 1 - \frac{1}{D})\ell_2(n)}.
$$
Choosing $\alpha > 0$ such that $H(\alpha) < \frac{1}{D}$, we find that with probability $\geq 1 - \exp\{-\Omega(\ell_2(n))\}$ over a uniform choice of $r_2$, $r_2$ is $\alpha \ell_2(n)$-far from any $r'$ such that $C(r', \cdot)$ is satisfiable.  For such $r_2$ and for any $z$, condition (\ref{lem:searchcsp}.i) of Lemma \ref{lem:searchcsp} tells us that $\Val_{\psi}(r_2, z) < 1 - \Omega(\alpha)$.  

Thus if we fix $\eps > 0$ as an appropriately small constant and choose an appropriate $s(|r_2|) = \exp\{-\Omega(|r_2|)\}$, we have $\psi \in \Pi_{AM-CSP, NO}$.  This completes the proof of correctness for our reduction.
\end{proof}

\section{Randomized Optimization Hypotheses Imply Collapse of $\AM$}\label{derandsec}

What significance might Theorem~\ref{thm:amkcsp}, our `PCP characterization of $\AM$', have for the project of trying to prove new upper bounds on the power of this class?  In the Introduction we gave two hypotheses inspired by Theorem~\ref{thm:amkcsp}.  Each of these hypotheses, if true, would have major implications for the study of $\AM$; this is the content of Claims~\ref{clm:hypA} and~\ref{clm:hypB} from the Introduction, which we prove next.

\begin{proof}[Proof of Claim~\ref{clm:hypA}]
Let $L \in \AM$; then $(L, \overline{L}) \in \prAM$.  Given an instance $x$, let Arthur run the reduction in Theorem~\ref{thm:amkcsp} on input $x$, producing a 2-CSP $\psi(r, z)$.  Let Merlin send Arthur a polynomial-sized circuit $C: \{0,1\}^{\ell} \rightarrow \{0,1\}^m$, with $\delta := \eps$ (here $\eps$ is from Theorem \ref{thm:amkcsp}).  Then Arthur runs $C$ on a sufficiently large ($O(\poly(n))$) number of random choices of $r$, accepting only if he finds an $r$ such that $\Val_{\psi}(r, C(r)) \geq 1 - \eps$.  

First suppose $x \in L$; then by the completeness property of our reduction, for all $r$ there exists a $z$ for which $\Val_{\psi}(r, z) = 1$.  If Merlin sends the circuit $C_{\psi}$ assumed to exist by Hypothesis A, then with at least $1/\poly(|\ell|)$ probability over $r$, $\Val_{\psi}(r, C(r)) \geq 1 - \eps$.  So if Arthur samples a sufficiently large (polynomial) number of strings $r$, Arthur will accept with probability $> 2/3$.

Next suppose $x \notin L$; then our reduction guarantees that for all but an exponentially small fraction of strings $r$, for all $z$ $\Val_{\psi}(r, z) < 1 - \eps$.  So Arthur's acceptance probability is negligible no matter what circuit Merlin sends.
Thus we have an $\MA$ protocol for $L$. \end{proof}

\begin{proof}[Proof of Claim~\ref{clm:hypB}]
We apply Hypothesis B with $\delta := \eps/3$, yielding a value $t = t(\delta)$.  Let $L \in \AM$ be given, and let Arthur run the reduction from Theorem~\ref{thm:amkcsp} on input $x$, yielding an instance $\psi(r, z)$.
Let Merlin send a description of a function $F(r)$, where each output of $F$ depends on at most $t$ bits of $r$ (note $F$ can be described in polynomial size).  Arthur performs explicit variable-substitutions $z = F(r)$ in $\psi$ and uses linearity of expectation to exactly compute $\bE_r[\Val_{\psi}(r, F(r))]$.  

If $x \in L$ and Merlin sends $F_{\psi}$ as given by Hypothesis B, this expectation is at least $(1 - \delta)^2 > 1 - 2\eps/3$.  On the other hand, if $x \notin L$ then, regardless of the function sent, this expectation is at most $(1 - \eps) + \exp\{-\Omega(|r|)\}$.  Thus for $|r|$ large enough we can distinguish the two cases.  (If $|r|$ is below a fixed threshold, Arthur can instead request that Merlin send optimal values $z(r)$ for each $r$.)  Arthur's computations are deterministic and polynomial-time, so the above defines an $\NP$ protocol for $L$.
\end{proof}

Note that Claim~\ref{clm:hypB} would hold even if we weakened Hypothesis B, allowing each coordinate of $F(r)$ to depend on $t(\delta, n) = O_{\delta}(\log n)$ coordinates.  We state Hypothesis B in a stronger form because, although we believe it is false, we don't know how to disprove it unconditionally even in the form given.

\section{Evidence Against the Randomized Optimization Hypotheses}\label{nphardsec}

Next we use Lemma~\ref{lem:searchcsp}, in conjunction with known results about amplification of hardness, to prove Theorems~\ref{thm:avghard1} and~\ref{thm:avghard2}.  That is, under various complexity-theoretic assumptions, we exhibit families of 2-CSPs $\psi(r, z)$ for which it is hard on average to approximately optimize over $z$, for randomly chosen $r$.  As mentioned earlier, the conclusions of both Theorems are easily seen to falsify both of our Randomized Optimization Hypotheses, and we consider this evidence that these hypotheses are probably false.  



\vspace{.5 em}

First, amplification of hardness in $\NP\cap \coNP$ from $(1 - 1/\poly(n))$-hardness to $2/3$-hardness is made possible by the following result of Impagliazzo~\cite[essentially Thm. 2]{Imp}:
\begin{mythm}\label{thm:mild1}~\cite{Imp} Suppose that there exists a language $L$, a function $s(n)$, and a $c > 0$, such that $L$ is $(1 - \frac{1}{n^c})$-hard for size $s(n)$.  Then for any $c' > 0$, there exists another language $L'$ such that $L'$ is $(\frac{1}{2} + O(\frac{1}{n^{c'}}))$-hard for size $\frac{s(n)}{n^{O(1)}}$.  Moreover, $L'$ is polynomial-time truth-table reducible to $L$.

\end{mythm}

\begin{mylem}\label{lem:mild2} Suppose that there exists a language $L \in \NP \cap \coNP$ and $\gamma, c > 0$ such that $L$ is $(1 - \frac{1}{n^c})$-hard for size $2^{\gamma n}$.  Then there exists another language $L' \in \NP \cap \coNP$ and a $\gamma' > 0$ such that $L'$ is $2/3$-hard for size $2^{\gamma' n}$ (for sufficiently large $n$).

\end{mylem}

\begin{proof} Apply Theorem~\ref{thm:mild1}, with $s(n) := 2^{\gamma n}$ and with any $c' > 0$ and $\gamma' \in (0, \gamma)$, and use the fact that $\NP \cap \coNP$ is closed under polynomial-time reducibilities, i.e., $\Polytime^{\NP \cap \coNP} = \NP \cap \coNP$.
\end{proof}


Next, the Impagliazzo-Wigderson pseudorandom generator~\cite{IW} allows us to amplify `moderate' hardness of the type produced by Lemma \ref{lem:mild2} into `extreme' hardness, albeit of a function problem rather than a decision problem (in~\cite{IW} additional techniques are used to produce extremely hard decision problems, but we do not follow this path).  The next definition  follows~\cite{IW} (and previous works).  Given a language $L$, an integer $c \geq 1$, a parameter $k = k(n)$, and a function $G(r): \{0, 1\}^{cn} \rightarrow \{0, 1\}^{k \times n}$ (called a `generator' function), define $L^k \circ G: \{0, 1\}^{cn} \rightarrow \{0, 1\}^{k}$ by 
$$
(L^k \circ G) (r)  :=  \left( L(G_1(r)), L(G_2(r)), \ldots , L(G_k(r))    \right),
$$
where the string $G(r)$ is divided into $k$ blocks $G_1(r), \ldots, G_k(r)$, each of length $n$.  The basic idea is that if $G(r)$ is appropriately `pseudorandom', then the collection $G_1(r), \ldots, G_k(r)$ should behave in important respects like a truly independent collection of random strings.  In particular, if it is somewhat hard to compute $L(x)$ for a random $x$, it should be very hard to compute $(L^k \circ G) (r)$ correctly when $k$ is large.

The following result (a restatement of~\cite[Thm. 2.12]{IW}) gives the main hardness-amplification property of the generator defined in that paper, which we denote $G_{IW}$.  

\begin{mythm} \label{thm:iwmain}~\cite{IW} For any $\gamma > 0$, there are $\gamma', c > 0$, and a polynomial-time computable $G_{IW}: \{0, 1\}^{cn} \rightarrow \{0, 1\}^{n \times n}$, such that: if $L$ is $2/3$-hard for size $2^{\gamma n}$, then $(L^n \circ G_{IW}) (r)$ is \\$2^{-\gamma' n}$-hard for size $2^{\gamma' n}$.
\end{mythm}
(Recall our definition of average-case hardness for general functions from Section~\ref{basicsec}.)


\begin{proof}[Proof of Theorem~\ref{thm:avghard1}] We begin by applying Lemma~\ref{lem:mild2} to our language $L \in \NP \cap \coNP$, yielding a language $L' \in \NP \cap \coNP$ that is 2/3-hard for circuits of size $2^{\gamma_0 n}$ for some $\gamma_0 > 0$.  Then we apply Theorem~\ref{thm:iwmain} to $L'$; we derive a $\gamma' > 0$, such that $((L')^n \circ G_{IW}) (r)$ is  $2^{-\gamma' n}$-hard for size $2^{\gamma' n}$.  

Since $L' \in \NP \cap \coNP$, there exists a polynomial-time witness predicate $M(x, w)$, producing outputs from $\{0, 1, ?\}$, satisfying:
\begin{enumerate}
\item for all $(x, w), M(x, w) \in \{L'(x), ?\}$;
\item for all $x$, there exists a $w$ such that $M(x, w) = L'(x)$;
\item $|w| = O(\poly(n))$.
\end{enumerate}
Let $t(n) = |w|$.  We reformat $M$ if necessary to ensure that the first bit of $w$ consists of a `claim' bit, call it $w_{cl}$, such that for any $(x, w)$ with $M(x, w) = L'(x)$, we have $w_{cl} = L'(x)$.  Next we define $M'(x, w)$, which outputs 1 if $M(x, w) \in \{0, 1\}$, 0 otherwise.  $M'$ is also polynomial-time computable.

Define a predicate $Q(r, w_1, \ldots, w_n): \{0, 1\}^{cn} \times \{0, 1\}^{n \times t(n)} \rightarrow \{0, 1\}$ as follows: $Q(r, w_1, \ldots, w_n) = 1$ iff for all $i \in [n]$, $M'(G_{IW, i}(r), w_i) = 1$.  $Q$ is polynomial-time computable since $G_{IW}$ and $M'$ are, so let $Q_n$ be a $O(\poly(n))$-sized circuit for $Q$ on input parameter $n$.  Clearly $Q_n$ is efficiently constructible.

We claim that $Q$ defines a hard-on-average search problem. To see this, suppose $C(r): \{0, 1\}^{cn} \rightarrow \{0, 1\}^{n \times t(n)}$ is any circuit of size at most $2^{\gamma' n}$ which has some $p(n)$ probability over $r$ of outputting a collection $w_1, \ldots, w_n$ for which $Q(r, w_1, \ldots, w_n) = 1$.  Then we may construct a circuit $C'(r) \{0, 1\}^{cn} \rightarrow \{0, 1\}^n$ that simply restricts the output of $C(r)$ to the `claim' bits of the strings $w_1, \ldots, w_n$ that $C$ produces.  Observe that $C'(r)$ has a $p(n)$ chance (over $r$) of correctly outputting $((L')^n \circ G_{IW}) (r)$.  Moreover, $C'(r)$ also has size bounded by $2^{\gamma' n}$.  We conclude $p(n) \leq 2^{-\gamma' n}$.

We invoke Lemma~\ref{lem:searchcsp} with $\eps := \gamma'/(2c)$, yielding a poly-time algorithm $A$ (and an associated $\nu > 0$).  We apply this $A$ to $Q_n$, yielding a 2-CSP $\psi_n(r, z)$ (here $|z| = d(n) = O(\poly(n))$).
We claim that the 2-CSP family $\{\psi_n(r, z)\}_{n > 0}$ satisfies the conditions of Theorem~\ref{thm:avghard1}.  

To see this, first note that for all $r$, $M'(r, \cdot)$ is satisfiable; so, there exists $w_1, \ldots, w_n$ such that $Q_n(r, w_1, \ldots, w_n) = 1$.  Thus by condition (\ref{lem:searchcsp}.i) of Lemma~\ref{lem:searchcsp}, there exists $z$ such that $\Val_{\psi_n}(r, z) = 1$.  So condition (\ref{thm:avghard1}.i) is satisfied.

To establish condition (\ref{thm:avghard1}.ii), let $\gamma_2 := \eps$ and $\theta := \nu$.  Suppose $C(r'): \{0, 1\}^{cn} \rightarrow \{0, 1\}^{|w'|}$ is a circuit of size at most $2^{\gamma_2 n}$, such that with some probability $q(n)$, $\Val_{\psi_n}(r', C(r')) > 1 - \theta$.  By condition (\ref{lem:searchcsp}.ii) of Lemma~\ref{lem:searchcsp}, there exists a circuit $\tilde{C}(r): \{0, 1\}^{cn} \rightarrow \{0, 1\}^{n \times t(n)}$, such that with probability at least $q(n) \cdot 2^{-\eps (cn)}$ over $r$, $Q_n(r, \tilde{C}(r)) = 1$.  Moreover, $\tilde{C}$ is of size at most $|C| + O(\poly(n))$, which for large enough $n$ is less than $2^{\gamma' n}$.  By our previous analysis we find that $q(n) \cdot 2^{-\eps (cn)} \leq  2^{-\gamma' n}$, i.e., $q(n) \leq 2^{(\eps c - \gamma')n} = 2^{-\gamma'n/2}  = \exp\{-\Omega(n)\}$.  We have proved condition (\ref{thm:avghard1}.ii), for our settings of $\gamma_2, \theta$.
\end{proof}

Next we turn to prove Theorem~\ref{thm:avghard2}.  For this, we need a more detailed analysis of the generator $G_{IW}$.  Besides the hardness-amplification property summarized in Theorem~\ref{thm:iwmain}, $G_{IW}$ has another useful property: with very high probability over $r$, the fraction of the strings $G_{IW, 1}(r), \ldots , G_{IW, n}(r)$ which are in $L$ is close to $|L \cap \{0, 1\}^n|/2^n$, that is, close to the fraction we'd expect if these strings were drawn independently and uniformly.  To prove this fact (not proved or used in~\cite{IW}), we first describe the generator in more detail.

The input $r$ to $G_{IW}$ consists of two parts, $r = (r_a, r_b)$.  $G_{IW}(r_a, r_b)$ is defined blockwise for $i \in [n]$ as 
$$
G_{IW, i}(r_a, r_b) = K_i(r_a) + K'_i(r_b),
$$
with $K_i: \{0, 1\}^{|r_a|} \rightarrow \{0, 1\}^n$, $K'_i: \{0, 1\}^{|r_b|} \rightarrow \{0, 1\}^n$, and with $+$ denoting bitwise addition mod 2.
The definition of $K'_i$ is not important to us; let us describe the functions $K_i$.  The string $r_a$ defines a random walk of length $n$ (counting the starting vertex) on an explicit expander graph $\mathcal{G}_n$ with vertex set $\{0, 1\}^n$.  $\mathcal{G}_n$ is $16$-regular with normalized second eigenvalue $\lambda_n$ at most some fixed $\lambda < 1$.  $v_i = K_i(r_a)$ represents the $i$-th vertex visited in this walk. $v_1$ is a uniform element, and each subsequent step $v_{i+1}$ is a uniform choice from among the neighbors of $v_i$.  (Note that this can be achieved with $|r_a| = O(n)$ random bits as claimed.)

We will use the following powerful result, called the Strong Chernoff Bound for Expander Walks, proved by Healy~\cite{Hea}.  

\begin{mythm}\label{thm:chernexp}\cite{Hea} Let $G = (V, E)$ be a $d$-regular graph with second eigenvalue $\lambda$, let $m > 0$, and let $f_1, \ldots, f_m: V \rightarrow [0, 1]$ have expectations $\mu_1, \ldots, \mu_m$ (over a uniform choice of $v \in V$).  Taking a random walk $v_1, \ldots, v_m$ on $G$ with uniform starting-point, we have for all $\eps > 0$,
$$
\Pr \left[  \left|     \sum_{i \leq m} f_i(v_i)     - \sum_{i \leq m} \mu_i    \right|  \geq \eps m   \right]     \leq 2e^{-\frac{\eps^2 (1 - \lambda)m}{4}} .  
$$
\end{mythm}

A more general claim was made earlier by Wigderson and Xiao~\cite{WX05}, but the proof contained an error, as pointed out in~\cite{WX08}.  (A valid proof of the Theorem above, with different constants, can still be extracted from~\cite{WX05}.)

For any $r \in \{0, 1\}^{cn}$, let $\sharp (r) := |\{i \in [n]: L(G_{IW, i}(r)) = 1\}|$.  Theorem~\ref{thm:chernexp} implies the following concentration bound for the generator $G_{IW}$:

\begin{mylem}\label{lem:genconc}  Let $L$ be an arbitrary language.  Let $c_n = |L \cap \{0, 1\}^n|/2^n$.  Then for any fixed $\delta > 0$,
$$
\Pr_{r}[  \left| \sharp (r)    -  c_n\cdot n     \right|  \geq \delta n    ]    \leq \exp\{-\Omega(n)\}. 
$$
\end{mylem}

\begin{proof} Recall that $r = (r_a, r_b)$.  We show that the above inequality is true after conditioning on any value of $r_b$; this will prove the Lemma.  Let $(v'_1, \ldots, v'_n) = (K'_1(r_b), \ldots, K'_n(r_b))$.  Then for $i \in [n]$, $G_{IW, i}(r) \in L$ iff $K_i(r_a) + v'_i \in L$, or equivalently $K_i(r_a) \in L_n + v'_i$ (where $L_n :=  L \cap \{0,1\}$ and $L_n + v'_i = \{x + v_i: x \in L_n\}$).  

Define $f_i: \{0, 1\}^n \rightarrow \{0, 1\}$ to be the characteristic function of $L_n + v'_i$.  Clearly $\mu_i = c_n$ for all $i$.  The result now follows by a direct application of Theorem~\ref{thm:chernexp}, using the fact that $\mathcal{G}_n$ has second eigenvalue bounded away from 1.
\end{proof}

\begin{proof}[Proof of Theorem~\ref{thm:avghard2}]   Our choice of $\eps_0$, determined later, will be no larger than $1/6$, so by Theorem~\ref{thm:iwmain}, our hypothesis implies there is a $\gamma' > 0$, a $c > 0$, and a polynomial-time $G_{IW}: \{0, 1\}^{cn} \rightarrow \{0, 1\}^{n \times n}$, such that: for any circuit $C: \{0, 1\}^n \rightarrow \{0, 1\}^n$ of size at most $2^{\gamma' n}$,
$$
\Pr_r [C(r) = (L^n \circ G_{IW}) (r)     ]       \leq 2^{-\gamma' n} .
$$
Let $M(x, w)$ be a polynomial-time verifier for $L$: $x \in L$ iff there exists $w$ such that $M(x, w) = 1$.  Let $t(n) = |w| = O(\poly(n))$.

Let $\sharp(r)$ be as defined after Theorem~\ref{thm:chernexp}. If $C: \{0, 1\}^{cn} \rightarrow \{0, 1\}^{t(n) \times n}$ is a circuit producing $n$ strings $w_1, \ldots, w_n$, each of length $t(n)$, define
$$
\sharp_C (r) := |\{i \in [n]: M(G_{IW, i}(r), w_i) = 1\}|.
$$

\begin{myclm}\label{subclm} There exists $\gamma'' > 0$ such that the following holds.  If $C(r): \{0, 1\}^{cn} \rightarrow \{0, 1\}^{t(n) \times n}$ is a circuit of size at most $2^{\gamma'' n}$, then
$$
\Pr_{r}[  \sharp (r) < \sharp_C (r) + \gamma'' n      ]  < 2^{-\gamma'' n} . 
$$
\end{myclm}

\begin{proof}[Proof (of Claim~\ref{subclm})] Let $\alpha > 0$.  Say $C(r): \{0, 1\}^{cn} \rightarrow \{0, 1\}^{t(n) \times n}$ is a circuit of size at most $2^{\alpha n}$, such that $\Pr_{r}[  \sharp (r) < \sharp_C (r) + \alpha n      ]  \geq 2^{-\alpha n}$.  Consider the following randomized procedure that attempts to compute $(L^n \circ G_{IW}) (r)$:
\begin{itemize}
\item Let $C(r) = (w_1, \ldots, w_n) \in \{0, 1\}^{t(n) \times n}$.  Let $I \subseteq [n]$ be the indices $i$ for which $M(G_{IW, i}(r), w_i) = 1$.  Pick a random subset $J$ of $[n]$, uniformly from the set of all subsets of size less than $\alpha n$ (including the empty set). Output the characteristic vector of $I \cup J$.
\end{itemize}
We analyze this procedure.  Suppose $r$ is any input for which $\sharp (r) < \sharp_C (r) + \alpha n$.  Note that by definition of $M$, we always have $I \subseteq \{i \in [n]: L(G_{IW, i}(r)) = 1\}$.  Then there exists a $J \subseteq [n] \setminus S$, of size less than $\alpha n$, such that $I \cup J =  \{i \in [n]: L(G_{IW, i}(r)) = 1\}$.  Thus conditioned on this event, the procedure succeeds with probability at least $\frac{1}{V_{n, \alpha n}} \geq 2^{-H(\alpha) n}$.  So the overall success probability is at least $2^{-\alpha n}\cdot 2^{-H(\alpha) n} = 2^{-(\alpha + H(\alpha)) n}$.  

Let us nonuniformly fix a setting $J$ that maximizes the procedure's success probability, and use this choice to run the procedure.  The result is a (nonuniform) circuit of size $2^{\alpha n} + O(\poly(n))$, with success probability $\geq 2^{-(\alpha + H(\alpha)) n}$.  For $\alpha$ sufficiently small this contradicts the hardness of $(L^n \circ G_{IW})$, proving the claim.
\end{proof}

Now we set $\eps_0 := \min (1/6, \gamma''/4)$.  Fix any circuit $C: \{0, 1\}^{cn} \rightarrow \{0, 1\}^{t(n) \times n}$ of size at most $2^{\gamma''n}$.
We use Lemma~\ref{lem:genconc} applied to $\delta := \eps_0$, and the previous Claim, to find that, with probability $\geq 1 - \exp\{-\Omega(n)\}$ over $r$, we have the simultaneous inequalities $(c_n + \eps_0)n > \sharp (r) > (c_n - \eps_0)n$ and $\sharp (r) \geq \sharp_C (r) + \gamma'' n$.  Call a string $r$ with this property \emph{$C$-typical}.

What is $c_n$?  We claim it must lie in $[1/2 - \eps_0, 1/2 + \eps_0]$.  For otherwise, a size-1 circuit could guess $L(x)$ with probability greater than $1/2 + \eps_0$ by guessing the majority value on length $n$, contrary to our hardness assumption about $L$.  Thus for a $C$-typical $r$, $\sharp_C (r) \leq (1/2 + \eps_0)n - \gamma''n < (1/2 - 3\gamma''/4)n$ and also $\sharp(r) \geq (1/2 - \eps_0)n - \eps_0 n \geq (1/2 - \gamma''/2)n$.  

Defining $\eta$ as some rational number in the interval $(1/2 - 3\gamma''/4, 1/2 - \gamma''/2)$, define a predicate $Q(r, w_1, \ldots, w_n): \{0, 1\}^{cn} \times \{0, 1\}^{n \times t(n)} \rightarrow \{0, 1\}$ as follows: $Q(r, w_1, \ldots, w_n) = 1$ iff for at least an $\eta $ fraction of indices $i$ we have $M(G_{IW, i}(r), w_i) = 1$.  $Q$ is itself polynomial-time computable, computed by some uniform family $\{Q_{n}\}_{n > 0}$ of poly-size circuits.  We have the key property that for a $C$-typical $r$, there exist $w_1, \ldots, w_n$ such that $Q(r, w_1, \ldots, w_n) = 1$, yet $Q(r, C(r)) = 0$.

Invoke Lemma~\ref{lem:searchcsp} with $\eps := \gamma''/(2c)$, yielding an algorithm $A$ (and an associated $\nu > 0$).  Then we claim $\{A(Q_n)\}_{n > 0} = \{\psi_n(r, z)\}_{n > 0}$ is the desired family of 2-CSPs (here $|r| = cn,  |z| = d(n) = O(\poly(n))$).    First we verify condition (\ref{thm:avghard2}.i).  Consider any $r$ for which there exists a $w_1, \ldots, w_n$ such that $Q_n(r, w_1, \ldots, w_n) =1$.   By condition (\ref{lem:searchcsp}.i) of Lemma \ref{lem:searchcsp}, we find that in this case there exists $z$ such that $\Val_{\psi_n}(r, z) = 1$.  Since all but an $\exp\{-\Omega(n)\}$ fraction of $r$ have this property, condition (\ref{thm:avghard2}.i) is satisfied.

To establish condition (\ref{thm:avghard2}.ii), fix $\gamma_2$ as any value in $(0, \gamma'')$ and let $\theta := \nu$.  Suppose $C(r): \{0, 1\}^{cn} \rightarrow \{0, 1\}^{d(n)}$ is a circuit of size at most $2^{\gamma_2 n}$, such that with some probability $q(n)$, $\Val_{\psi_n}(r, C(r)) > 1 - \theta$.  By condition (\ref{lem:searchcsp}.ii) of Lemma \ref{lem:searchcsp}, there exists a circuit $\tilde{C}(r): \{0, 1\}^{cn} \rightarrow \{0, 1\}^{n \times t(n)}$, such that with probability at least $q(n) \cdot 2^{-\eps (cn)}$ over $r$, $Q_n(r, \tilde{C}(r)) = 1$.  Note that such an $r$ fails to be $\tilde{C}$-typical.  Moreover, $\tilde{C}$ is of size at most $|C| + O(\poly(n))$, which for large enough $n$ is less than $2^{\gamma'' n}$.  

So, by our previous analysis we find that $q(n) \cdot 2^{-\eps cn} \leq  2^{-\gamma'' n}$, i.e., $q(n) \leq 2^{(\eps c - \gamma'')n} = 2^{-\gamma''n/2}  = \exp\{-\Omega(n)\}$.  We have proved condition (\ref{thm:avghard2}.ii).  This completes the proof of Theorem~\ref{thm:avghard2}.
\end{proof}

Finally, we note that versions of Theorems~\ref{thm:avghard1} and~\ref{thm:avghard2} can be proved, in which both the hypotheses and conclusions apply, not to general circuits, but to the class of $\TCzero$ circuits (i.e., constant-depth Boolean circuits with majority gates), or any circuit class containing $\TCzero$.  This is because all the reductions involved can be carried out in $\TCzero$.  (For a discussion of why the Impagliazzo and Impagliazzo-Wigderson constructions amplify hardness in $\TCzero$, see Agrawal~\cite{Agr}; the difficulties in amplifying hardness in lower classes like $\ACzero$ were explored by Shaltiel and Viola~\cite{SV}.)

\section{Questions for Further Research}

\begin{itemize}

\item Does our approximation problem remain $\prAM$-complete if each variable in the CSP $\psi(r, z)$ is restricted to appear in only a constant number of constraints?  (The `expander-replacement' technique~\cite{PapYan, Pap} allows us to restrict the occurrences of $z$-variables in our $\prAM$-completeness proof; it is the `stochastic' $r$-variables which pose a challenge.)  Alternatively, can one perhaps show that under this restriction the problem lies in $\NP$?  


\item Can we unconditionally disprove Hypothesis B?  Given the sharp limitations of $\NCzero$ circuits this might be possible.

\item Can PCP ideas be used to give new upper bounds on the class $\AM$?

\item Find more applications of PCPPs in complexity theory.

\end{itemize}

\section{Acknowledgements}

I thank Scott Aaronson, Madhu Sudan, and some anonymous referees for helpful comments.

\bibliographystyle{alpha}
\newcommand{\etalchar}[1]{$^{#1}$}

\end{document}